\documentclass[journal,onecolumn]{IEEEtran}

\usepackage{amssymb,bm,amsthm,color,amsmath,url,multirow}
\usepackage{diagbox}
\usepackage{graphicx}
\usepackage{subfigure}

\usepackage{cite}

\newtheorem{thm}{Theorem}[section]
\newtheorem{ques}{Question}[section]
\newtheorem{lem}{Lemma}[section]
\newtheorem{dfn}{Definition}[section]

\newtheorem{rmk}{Remark}[section]

\newtheorem{example}{Example}[section]
\newtheorem{clm}{Claim}
\newenvironment{pf}{{\noindent\it Proof:}}{\hfill $\blacksquare$\par}

\newcommand{\RNum}[1]{\lowercase\expandafter{\romannumeral #1\relax}}
\newcommand{\Rnum}[1]{\uppercase\expandafter{\romannumeral #1\relax}}

\begin{document}

\title{A generic framework for coded caching and distributed computation schemes}
\author{Min Xu, Zixiang Xu, Gennian Ge and Min-Qian Liu
	
\thanks{This project was supported by the National Key Research and Development Program of China under Grant 2020YFA0712100 and Grant 2018YFA0704703, the National Natural Science Foundation of China under Grant 11971325, Grant 12131001  and Grant
12231014,  Beijing Scholars Program, the Institute for Basic Science (IBS-R029-C4), and National Ten Thousand Talents Program.}
	
		\thanks{M. Xu ({\tt minxu0716@qq.com}) and M. Liu ({\tt mqliu@nankai.edu.cn}) are with the School of Statistics and Data Science, LPMC \& KLMDASR, Nankai University, Tianjin 300071, China.}
		\thanks{Z. Xu ({\tt zxxu8023@qq.com}) is with Extremal Combinatorics and Probability Group, Institute for Basic Science, Daejeon, South Korea.  This work was started when Z. Xu is a Ph.D student in Capital Normal University.}
		\thanks{G. Ge ({\tt gnge@zju.edu.cn}) is with the School of Mathematical Sciences, Capital Normal University,
			Beijing 100048, China.

	}}

\maketitle
\begin{abstract}
Several network communication problems are highly related such as coded caching and distributed computation.
The centralized coded caching focuses on reducing the network burden in peak times in a wireless network system and the coded distributed computation studies the tradeoff between computation and communication in distributed system.
In this paper, motivated by the study of the only rainbow $3$-term arithmetic progressions set, we propose a unified framework for constructing coded caching schemes. This framework builds bridges between coded caching schemes and lots of combinatorial objects due to the freedom of the choices of families and operations.
We prove that any scheme based on a placement delivery array (PDA) can be represented by a rainbow scheme under this framework and lots of other known schemes can also be included in this framework.
Moreover, we also present a new coded caching scheme with linear subpacketization and near constant rate using the only rainbow $3$-term arithmetic progressions set.
Next, we modify the framework to be applicable to the distributed computing problem. We present a new transmission scheme in the shuffle phase and show that in certain cases it could have a lower communication load than the schemes based on PDAs or resolvable designs with the same number of files.
\end{abstract}


\section{Introduction}
In recent years, wireless communication systems such as $5G$ are becoming more and more content-centric, while the scale of real world data is becoming larger, the cost of communication between remote users and servers relies on the size of data. There are several problems aiming at reducing the communication load in different applications such as coded caching and coded distributed computing.

The first problem we investigate is coded caching. In our daily life, wireless traffic has become a problem. One of the main driving factors for wireless traffic is the dramatic increase in demand for video content. Moreover, the high temporal variability of network traffic results in communication systems to be congested during peak-traffic times but under utilized during off-peak times. One approach to reduce peak traffic is to take advantage of memories distributed across the network to duplicate content. This duplication of content, called caching, is performed during off-peak times when network resources are abundant. During peak-traffic times, user demands can be served from these caches and the network congestion will be reduced. Hence, coded caching schemes are widely studied because of their applications in reducing the network burden and smoothing the network traffic.

The research of designing caching schemes to exploit the benefit of coding was initiated by Maddah-Ali and Niesen \cite{2014AliIT}. In their seminal work, they proposed the first centralized coded caching scheme, where the central idea is
to design an appropriate content placement and delivery strategy. From then on, the problem of designing centralized coded caching scheme is often called the coded caching problem.

Recently, several interesting methods were used to construct coded caching schemes. Yan et al. \cite{2017PDAIT} represented the placement delivery array (PDA) framework and the coded caching scheme they proposed has significantly lower subpacketizations than that of the Maddah-Ali and Niesen scheme in \cite{2014AliIT}. In \cite{2018ShangguanIT}, Shangguan et al. established a connection between coded caching schemes and $3$-partite $3$-uniform $(6,3)$-free hypergraphs and constructed coded caching schemes with constant rate and subpacketizations increasing sub-exponentially with the number of users. This connection was further expanded upon in terms of strong edge colorings of bipartite graphs by Yan et al. \cite{2018SEC}. In \cite{2017ISITrsgraph}, Shanmugam et al. showed coded caching with linear subpacketizations and near constant rate is possible using Ruzsa-Szemer\'{e}di graphs. There are also some other combinatorial approaches such as linear block codes~\cite{2017TANGLIISIT, 2018ITLinearBlock}, line graphs of bipartite graphs~\cite{2018Linegraph, 2019LinegraphISIT, 2019LinegraphISCIT, 2019Linegraph1902, 2020Linegraph} and combinatorial designs~\cite{2019ISITCD}.

The centralized coded caching schemes were also extended to a variety of more practical settings. For example, the decentralized coded caching was introduced in \cite{2015Decentralized}, where the coded delivery scheme is shown to achieve large gains in the rate, under a random or decentralized caching phase. The coded caching in popularity-based caching settings, online coded cachings and hierarchical coded cachings were also proposed in \cite{2018ITPopoularity, 2016Online1, 2017Online2, 2016ITHierarchical}. Moreover, several interesting settings such as D$2$D-network, distributed computing and cache-aided interference management \cite{2016ITD2DJi, 2018ITdistributed, 2017ITcacheaid} also take advantage of the constructions of centralized coded caching schemes. In a word, the design of centralized coded caching schemes is still of great importance.

Another important problem in reducing the communication load which focuses on distributed scenario is coded distributed computing. Due to the rapid growth of large-scale machine learning and big data analysis, the distributed system has been widely used in daily life. Storing data distributedly brings some advantages while increases the communication load when we need the data which is not stored locally. In real life, we not only need to read the data itself, but also need to get some special functions about the data.

There are some frameworks of function computing on a distributed data system, such as MapReduce \cite{2004deanmapreduce} and Spark \cite{2010zahariaspark}. Based on MapReduce, Li et al.~\cite{2018ITdistributed} provided a tradeoff between communication load and computation load and they constructed a scheme achieving the optimal communication load where the storage pattern is given. It is worth to note that the main idea of the scheme in \cite{2018ITdistributed} is similar with the MAN scheme in coded caching~\cite{2014AliIT}. Thus, the PDA scheme can also be applied in distributed computing while containing the MAN scheme as a special case~\cite{2020YanDistributed}. Another combinatorial object known as resolvable design was connected to a distributed computing scheme by Konstantinos et al.~\cite{2020KKRA}. They pointed out that the large number of small tasks in~\cite{2018ITdistributed} has detrimental effects on the performance of the scheme, which means though the MAN scheme attains the optimal communication load in theoretical prediction, it requires higher shuffling time. Therefore, it is reasonable to design a scheme which requires fewer files.

Our contribution can be concluded as follows.
\begin{enumerate}
    \item First, we propose a unified rainbow framework to yield centralized coded caching schemes. After coloring some elements with specific rules, we show that the uncolored elements can encode storage actions while elements receiving the same color encode delivery actions (XORs), which leads to a simple but very useful relationship between the coded caching scheme and these elements. Moreover, we use this new idea to review several existing works and discuss the way to improve the existing schemes.
    \item Next, we study the problem of constructing centralized coded caching schemes with low subpacketization level based on only rainbow $3$-term arithmetic progressions sets, which are proposed by Pach and Tomon \cite{2019PachRainbowAP}. We present a coded caching scheme with linear subpacketization and near constant rate. Moreover, we propose a new delivery scheme based on some results in index coding problem, which can further reduce the transmission load.
    \item At last, we apply the rainbow framework into coded distributed computing problem, and construct a new scheme which has a different shuffle phase. Comparing with the schemes based on PDAs or resolvable designs, our new scheme has a lower communication load with the same number of required files.
\end{enumerate}

The rest of this paper is organized as follows. In Section~\ref{pre}, we introduce two kinds of network communication problems known as coded caching and distributed computing, as well as the coloring problem which is the main tool of our new schemes. Motivated by the rainbow structure, we propose a generalized rainbow framework in Section~\ref{framework}. Surprisingly, we find out that any PDA scheme can be represented in the rainbow framework, and present several examples in Section~\ref{existingscheme}. In Section~\ref{rainbow}, we derive a rainbow scheme with new parameters based on the only rainbow $3$-term arithmetic progressions set. Next, we modify the rainbow framework to be applicable to distributed computing problem in Section~\ref{distributed}. At last, we conclude our main results and propose some open problems in Section~\ref{conclusion}.

\section{Preliminary}\label{pre}
In this section, we introduce two kinds of problems in network communication and the main tool we used to construct our new caching scheme.
\subsection{Coded caching}
The first problem we study is the coded caching problem, which was first investigated by Maddah-Ali et al. in 2014 \cite{2014AliIT}.
In this kind of problem, there is a central server with a library of $N$ files $\{W_1, W_2,\ldots, W_N\}$, each file is $F$ bits or can be partitioned into $F$ subfiles.
Suppose in this system there are $K$ users, each of which has a cache with size of $M$ files and requires one file in the library.
Suppose user $i$ requires file $W_{d_i}$, denote the demand vector as $\mathbf{d}=(d_1,\ldots,d_K)$.
Before the users sending their demands $\mathbf{d}$ to the server, the server fills up all the caches.
When the server knows the users' demands, it sends $X_{\mathbf{d}}$ according to the users' caches and their demands $\mathbf{d}$.
The server and the users are connected by an error-free shared link, that is, every message sent by the server can be seen by all users.
The transmission rate of this system is defined as
\begin{equation*}
  R=\max\limits_{\mathbf{d}\in N^K}\frac{|X_\mathbf{d}|}{F}.
\end{equation*}
The main purpose of coded caching is to design the placement of subfiles such that the transmission load $R$ for all possible users' demands is as small as possible.
In this paper we focus on the caching schemes with uncoded placement, that is, each user caches subfiles directly instead of functions of subfiles.

In \cite{2014AliIT}, a tradeoff between $K,M,N$ and $R$ was given as follows,
\begin{equation*}
  R^*\geq \max_{s\in\{1,\ldots,\min\{N,K\}\}}\bigg(s-\frac{s}{\lfloor N/s\rfloor}M\bigg).
\end{equation*}
They also presented a scheme with rate
\begin{equation*}
  R=K\bigg(1-\frac{M}{N}\bigg)\cdot\frac{1}{1+\frac{KM}{N}},
\end{equation*}
which was proved to be the optimal scheme under the uncoded placement \cite{2016Kai}. Although the rate is optimal, the subpaketization $F=\exp(K)$, which increases the complexity of the scheme. Various works about reducing the subpaketizations $F$ while keeping rate $R$ low have been done \cite{2019ISITCD,2018Linegraph, 2019LinegraphISIT, 2019LinegraphISCIT, 2019Linegraph1902, 2020Linegraph},\cite{2018ShangguanIT},\cite{2017ISITrsgraph,2017TANGLIISIT, 2018ITLinearBlock},\cite{2017PDAIT},\cite{2018SEC}. The most general scheme among these works is the placement delivery array (PDA)\cite{2017PDAIT}. It is convenient to use this array to represent the placement and the delivery phase of a coded caching scheme. Moreover, the connections between a coded caching scheme and other combinatorial objects were studied in recent years. Constructing a scheme with better performance is still an interesting problem.
\begin{figure}[h]
\centering
\includegraphics[width=10cm]{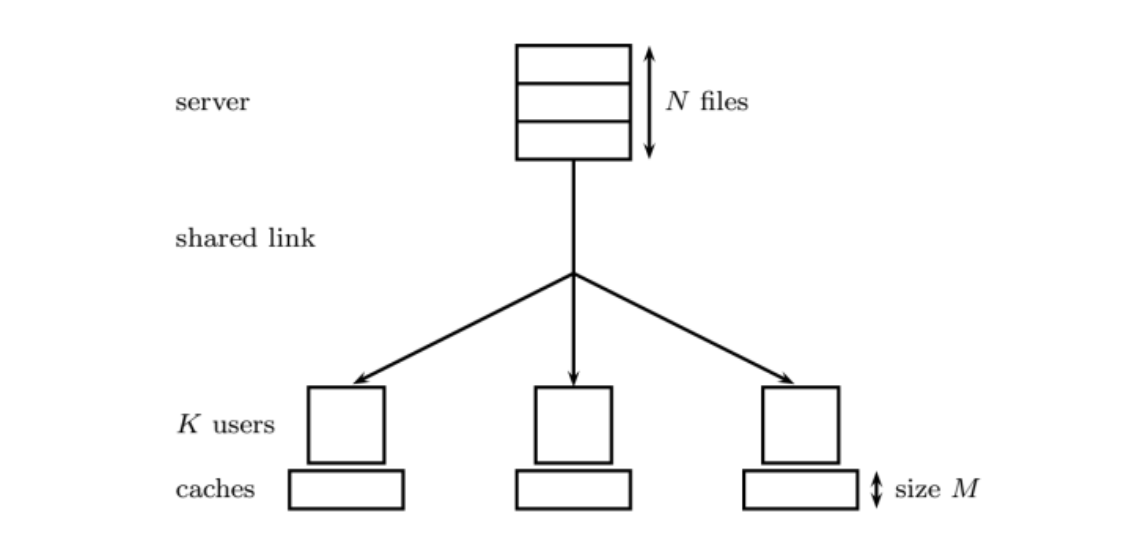}
\caption{Coded caching system}\label{picture:CCC}
\end{figure}

\subsection{Distributed computing}
Another problem considered in this paper is the distributed computation problem under the famous MapReduce framework \cite{2004deanmapreduce}.
In the MapReduce model, there is a master node who wants to compute $Q$ functions $\phi_1,\ldots,\phi_Q$ about $N$ files $W_1,\ldots,W_N$.
Suppose all files have the same size, and each function can be decomposed as map functions and reduce functions as
\begin{equation*}
  \phi_j(W_1,\cdots,W_N)=f_j(g_{j1}(W_1),\cdots,g_{jN}(W_N)),
\end{equation*}
where $\{g_{jn}:j\in[Q],n\in[N]\}$ are map functions and $\{f_j:j\in[Q]\}$ are reduce functions. Let $v_{jn}=g_{jn}(W_n)$ be the intermediate value of the computing task. Assume that each intermediate value is $T$ bits, i.e. $v_{jn}\in\mathbb{F}_2^T$.

The MapReduce model contains three phases: map, shuffle and reduce.
In the map phase, the master node assigns $N$ files to $K$ distributed work nodes. Each work node stores a subset of files $\mathcal{M}_k\subset\{W_i,i\in[N]\}.$
The computation load $r$ is the average number of each file stored in all work nodes, i.e.
\begin{equation*}
  r=\frac{\sum_{k\in[K]}|\mathcal{M}_k|}{N}.
\end{equation*}
Work node $k$ computes $\{v_{jn}:j\in[Q],n\in\mathcal{M}_k\}$, which means the node computes intermediate values based on the files it has.
In the shuffle phase, all work nodes need to communicate with each other to get the files they do not have. Each work node will be assigned $Q/K$ target functions. To finish the computation tasks, each node will send a function of intermediate values it has to other nodes based on the target function assignment. We use $X_k$ to denote the message sent by node $k$ and $|X_k|$ is the size of $X_k$. Then, the communication load $L$ is defined as
\begin{equation*}
  L=\frac{\sum_{k\in[K]}|X_k|}{QNT}.
\end{equation*}
In the reduce phase, after receiving all messages from other nodes, work node $k$ can get all intermediate values to finish the computation.

It's obvious that there exists a tradeoff between the communication load $L$ and computation load $r$, since if $r=K$, then $L=0$ and if $r$ becomes smaller, there are more intermediate values which are needed to be sent. Li et al. \cite{2018ITdistributed} studied this tradeoff and proposed a scheme named as CDC attaining the optimal tradeoff. The optimal tradeoff between $L$ and $r$ is characterized as
\begin{equation*}
  L^*=\frac{1}{r}(1-\frac{r}{K}).
\end{equation*}
There is a tight connection between the CDC scheme and the MAN scheme in coded caching, thus the CDC scheme also requires a large number of files, $N=\exp(K)$, which will highly increase the computation time in practice. Based on this observation, Yan et al. \cite{2020YanDistributed} used PDA to construct a distributed computing scheme which has a larger communication load and a smaller number of files. Almost at the same time, Konstantinos et al.\cite{2020KKRA} used another combinatorial object to construct a scheme with $L=\frac{1}{r-1}(1-\frac{r}{K})$ and $N={\rm{Poly}}(K)$. Constructing new scheme with better performance is still an interesting problem in this area.

\subsection{Only rainbow arithmetic progressions set}

We start with the definition of only rainbow arithmetic progressions set. For convenience, we usually write $k$-term arithmetic progressions as $k$-APs. We first introduce the concept of only rainbow $k$-APs set.
\begin{dfn}[Only rainbow $k$-APs set]
  Let $A$ be a subset of $[n],$ and $\Phi$ be a coloring function which colors every element of $A.$ Say $A$ is an only rainbow $k$-APs set if all $k$-APs in $A$ are rainbow, that is, each element in a $k$-AP receives distinct colors.
\end{dfn}
How many colors do we need to make sure that every $k$-AP is rainbow, that is, all of its elements receive distinct colors? For example, if $k=3,$ then at least $\frac{n}{2}$ colors are needed. Instead of coloring the whole set $[n],$ very recently, Pach and Tomon \cite{2019PachRainbowAP} considered the problem on dense subsets of $[n].$ They gave a surprising result that much fewer colors suffice if we do not insist on coloring all elements in $[n].$ More precisely, they showed the following result for $k=3$.
\begin{thm}[\!\!\cite{2019PachRainbowAP}]\label{thm:pach}
  Let $C$ be a sufficiently large integer and $n=C^{d}$ for some integer $d.$ There is a set $A\subseteq [n]$ with $|A|\geqslant n-n^{\alpha}$ and a coloring of $A$ with $n^{\beta}$ colors such that every $3$-AP in $A$ is rainbow, where $\alpha=1-\frac{1}{18C^{6}\log{C}}$ and $\beta=\log_{C}(10C^{\frac{16}{C}}\log_{2}{C}).$
\end{thm}
 For convenience, we call set $A$ an $(\alpha,\beta)$-only rainbow $3$-APs set if $A$ satisfies the properties in Theorem \ref{thm:pach}. Moreover, Theorem \ref{thm:pach} can be extended to longer arithmetic progressions easily. For more details, we refer the readers to Concluding remarks in\cite{2019PachRainbowAP}.

\section{Generalized rainbow framework}\label{framework}
We will describe a unified framework for coded caching schemes with uncoded placement. Let $\mathcal{A}$ and $\mathcal{B}$ be two collections with $K=|\mathcal{A}|$ and $F=|\mathcal{B}|.$ Let $\mathbb{\biguplus}$ be a certain operation, for example, the operation $\mathbb{\biguplus}$ can be a simple addition operation or a set union operation. Then we define the set $\mathcal{C}$ as
\begin{equation*}
 \mathcal{C}=\mathcal{A}\biguplus\mathcal{B}=\{a\biguplus b:a\in\mathcal{A},b\in\mathcal{B}\}.
\end{equation*}
 For a carefully selected property $\sigma$, we call any subset in $\mathcal{C}$ with property $\sigma$ as a $\sigma$-type structure.
 Consider a subset $\hat{\mathcal{C}}\subseteq\mathcal{C}$ and color every element in $\hat{\mathcal{C}}$ by a coloring function $\phi$.
 Call a subset $\hat{\mathcal{C}}\subseteq \mathcal{C}$ only rainbow $\sigma$-type set if every $\sigma$-type structure in $\hat{\mathcal{C}}$ is rainbow under the coloring function. For example, when $\mathcal{A}=\mathcal{B}=[m]$ and the operation is addition, the $\sigma$-type structure can be $3$-term arithmetic progressions.

\begin{dfn}\label{def:generalizedScheme}
  Let $K=|\mathcal{A}|$ and $F=|\mathcal{B}|,$ define $\mathcal{C}=\mathcal{A}\biguplus\mathcal{B}.$ Suppose there is a $\sigma$-type structure and a suitable subset $\hat{\mathcal{C}}\subseteq \mathcal{C}$ such that we can find a coloring function $\phi$ which makes every $\sigma$ structure rainbow in $\hat{\mathcal{C}},$ and we are able to construct a coloring function $\Phi$ over all the pairs $(a,b)$ with $a\biguplus b\in\hat{\mathcal{C}}$ based on $\phi$.
  Then we describe the placement phase and delivery phase with assistance of the above colored subset $\hat{\mathcal{C}}.$
  \begin{enumerate}
    \item \textbf{Placement phase:} For uncolored elements $a\biguplus b$ in $\mathcal{A}\biguplus \mathcal{B},$ user $a$ caches the $b$-th packet of all files in the library.
    \item \textbf{Delivery phase:} The delivery is based on the coloring function $\Phi$. Suppose there are $s$ elements~
$(a_{1},b_{1}),(a_{2},b_{2}),\ldots,$ $(a_{s}, b_{s})$ receiving the same color $c$ from $\Phi$, then we denote the following XOR multiplexing of packets as $W_c$
          \begin{equation*}
            W_c=\bigoplus\limits_{1\leqslant i\leqslant s} W_{d_{a_{i}}}^{(b_{i})}.
          \end{equation*}
          For each uncached pair $(a,b)$, define a constant
          \begin{equation*}
            m(a,b)=\#\{\Phi(a', b'):\ a'=a\ or\ b'\biguplus a\in\hat{\mathcal{C}}\},
          \end{equation*}
          and $m=\max\{\{m(a,b):a\cup b\in\hat{\mathcal{C}}\}\cup\{|\Phi|-1\}\}$.
          Let $P$ be an $m\times|\Phi|$ maximum distance separable (MDS) matrix.
          During the delivery phase, the server sends
          \begin{equation*}
            P\cdot(W_{c_1}, W_{c_2}, \cdots, W_{c_{|\Phi|}})^T.
          \end{equation*}
          Thus, the delivery phase consists of $m$ packet transmissions.
  \end{enumerate}
\end{dfn}

We present an example as follows.
\begin{example}\label{egcaching}
  Let $\mathcal{A}=[4]$, $\mathcal{B}=\{12,23,34,41\}$ and the operation be set union operation $\cup$. Therefore, $\mathcal{C}=\mathcal{A}\cup\mathcal{B}$ and assume $\hat{\mathcal{C}}=\{123,124,134,234\}.$  The coloring function $\phi$ on $\hat{\mathcal{C}}$ satisfies the rainbow property such that any three elements in $\hat{\mathcal{C}}$ receive distinct colors, which implies there are exactly $4$ colors. Using $\phi$, we can define the coloring function $\Phi$ over all the pairs $(a,b)$ such that $\Phi(a,b)=\phi(a\cup b)$.

  During the delivery phase, we need to count a special number before sending messages. For any pair $(a,b)\in\mathcal{A}\times\mathcal{B}$ such that $a\cup b\in\hat{\mathcal{C}}$, define $m(a,b)=\#\{\Phi(a',b'):\ a'=a\ or\ b'\cup a\in\hat{\mathcal{C}}\}$ and $m=\max_{\{(a,b):a\cup b\in\hat{\mathcal{C}}\}}m(a,b)$. In this example, $m=3$. Next, find a $3\times 4$ MDS array such as
  \begin{equation*}
      P=\begin{pmatrix}
      1 & 0 & 0 & 1 \\
      0 & 1 & 0 & 1 \\
      0 & 0 & 1 & 1
      \end{pmatrix},
  \end{equation*}
  and define the sum of subfiles corresponding to the same color $c_i$ as
  \begin{equation*}
      W_{c_i}=\bigoplus\limits_{\Phi(A_i\cup B_i)=c_i}W_{d_{a_i}}^{(b_i)}.
  \end{equation*}
  Then we send
  \begin{equation*}
      P\cdot (W_{c_1}, W_{c_2}, W_{c_3}, W_{c_4})^T.
  \end{equation*}
  More precisely, suppose $\phi(123)=c_1$, $\phi(124)=c_2$, $\phi(134)=c_3$, $\phi(234)=c_4$, then we send
  \begin{align*}
      & W_{d_1}^{(23)}\bigoplus W_{d_3}^{(12)}\bigoplus W_{d_2}^{(34)}\bigoplus W_{d_4}^{(23)}, \\
      & W_{d_2}^{(14)}\bigoplus W_{d_4}^{(12)}\bigoplus W_{d_2}^{(34)}\bigoplus W_{d_4}^{(23)}, \\
      & W_{d_1}^{(34)}\bigoplus W_{d_3}^{(14)}\bigoplus W_{d_2}^{(34)}\bigoplus W_{d_4}^{(23)}.
  \end{align*}
  We can check that each user can decode all the subfiles he requires from the transmission.
\end{example}

\begin{rmk}\label{rmk:Generalized}
  Given user set $\mathcal{A}$ and packets set $\mathcal{B},$ the first task in our generalized rainbow framework is to define the suitable operation $\biguplus,$ and then the most important thing is to give the appropriate $\sigma$-type structure. Then it turns to be a combinatorial problem, that is, we should design a coloring function to make sure every $\sigma$-type structure in $\hat{\mathcal{C}}$ is rainbow, with colors as few as possible. Trivially, we can color every element of $\hat{\mathcal{C}}$ using different colors, which implies that arbitrary structure in $\hat{\mathcal{C}}$ is rainbow.
\end{rmk}

\begin{rmk}
  In the delivery phase, we define a constant $m$ and find an $m\times |\Phi|$ MDS array. To make sure that such an array exists over $\mathbb{F}_2$ (which is always considered in coded caching problem), $m$ cannot be smaller than $|\Phi|-1$. However, if we assume that the computation between subfiles is over $\mathbb{F}_q$ for sufficiently large $q$, then we don't need $m\geq |\Phi|-1$, which means we can further reduce the transmission rate.
\end{rmk}

The main idea of the delivery scheme comes from the famous index coding problem.
The index coding problem can be described by a directed graph $G_d$ with $n$ vertexes, in which every vertex $x_i$ represents a user who requires $x_i$, and every directed edge from $x_i$ to $x_j$ means user $i$ has $x_j$ as side information. The index coding problem asks how many bits the server needs to broadcast such that each user can get the file he requires. For such a problem, there is a scheme with transmission rate $R=\chi_l(\bar{G}_d)$ in \cite{2013shanmugamlocal}, where $\chi_l(\bar{G}_d)$ is the local chromatic number of the complementary graph $\bar{G}_d$. In the directed graph $\bar{G}_d$, denote the closed out-neighborhood of a given vertex i as
\begin{equation*}
  N^{+}(i)=\{j\in V(\bar{G}_d):(i,j)\in E(\bar{G}_d)\}\cup\{i\}.
\end{equation*}
\begin{dfn}[\!\!\cite{2013shanmugamlocal}]
  The local chromatic number of a directed graph $\bar{G}_d$ is the maximum number of colors in any out-neighborhood minimized over all proper colorings of the undirected graph obtained by ignoring the orientation of the edges in $\bar{G}_d$, i.e.
  \begin{equation*}
    \chi_l(\bar{G}_d)=\min\limits_{c}\max\limits_{i\in V}|c(N^+(i))|.
  \end{equation*}
\end{dfn}

\begin{rmk}
  From above definition, it is obvious that the constant $m$ in the delivery scheme is the local chromatic number of the corresponding index coding problem.
\end{rmk}

Fix $a\in\mathcal{A}$, denote $Z_a$ as the number of $b\in\mathcal{B}$ such that $a\biguplus b\in\hat{\mathcal{C}}$. In this paper, we only consider that each user has the same cache size, i.e. $Z_a=Z$ for all $a\in\mathcal{A}$.
From the perspective of index coding, by selecting the $\sigma$-type carefully, we can get a proper coloring of $\bar{G}_d$. Therefore, we have the following result.

\begin{thm}
  The generalized rainbow framework provides a coded caching scheme with $(K=|\mathcal{A}|,F=|\mathcal{B}|,\frac{M}{N}=1-\frac{Z}{F},R=\frac{m}{|\mathcal{B}|})$.
\end{thm}
\begin{pf}
  The proof is directly from the above analysis.
\end{pf}

\section{Existing scheme under rainbow framework}\label{existingscheme}
In this section, we will highlight several existing works on coded caching schemes via different ideas and combinatorial objects. However, as we have discussed in Remark \ref{rmk:Generalized}, in our generalized rainbow framework, many existing schemes are equipped with the trivial coloring function. Hence we represent some of them and discuss the way to improve the existing schemes.

\subsection{All PDA schemes are rainbow schemes}
For a bipartite graph, a strong edge coloring function is a coloring function such that any two edges which have the same color are not adjacent and can not be connected by another edge.
The relationship between PDA schemes and strong edge coloring of a bipartite graph is studied in \cite{2018SEC}. The following result is known.
\begin{thm}[\!\cite{2018SEC}]
  Any $F\times K$ array $P$ is a PDA if and only if its corresponding edge colored bipartite graph $G(\mathcal{F}\cup\mathcal{K},\mathcal{E})$ satisfies
  \begin{enumerate}
    \item the vertex in $\mathcal{K}$ has a constant degree;
    \item the corresponding coloring is a strong edge coloring.
  \end{enumerate}
\end{thm}

Given a bipartite graph $G$ with vertex set $\mathcal{F}\cup\mathcal{K}$ and edge set $\mathcal{E}$. Then, in our new framework, let $\mathcal{A}=\mathcal{K}$, $\mathcal{B}=\mathcal{F}$ and the operation be the Cartesian product, i.e.
\begin{equation*}
  \mathcal{C}=\mathcal{A}\times\mathcal{B}=\{(a,b): a\in\mathcal{A},b\in\mathcal{B}\}.
\end{equation*}
Note that $\mathcal{C}$ is the set of all possible edges in $G$. Choose $\hat{\mathcal{C}}=\mathcal{E}$, and define a coloring function $\phi$ on $\hat{\mathcal{C}}$ such that
\begin{enumerate}
  \item If $(a_1,b_1),(a_2,b_2)\in\hat{\mathcal{C}}$ and $a_1=a_2$ or $b_1=b_2$ then $\phi((a_1,b_1))\neq\phi((a_2,b_2)).$
  \item If $(a_1,b_1),(a_2,b_2),(a_3,b_3)\in\hat{\mathcal{C}}$ and $|\{a_1,a_2,a_3\}|\leq2$, $|\{b_1,b_2,b_3\}|\leq2$ then $(a_1,b_1),(a_2,b_2)$ and $(a_3,b_3)$ are rainbow.
\end{enumerate}

In fact, any strong edge coloring function satisfies above conditions. It is easy to check that the first condition means that any two edges with the same color are not adjacent and the second condition is equivalent to that any two edges with the same color cannot be connected by another edge. Then we can select the coloring function $\Phi$ in the rainbow framework exactly equal to $\phi$. If we do not consider the constant $m$ and send subfiles according to their colors directly, we obtain a PDA scheme. From above analysis, we conclude the following theorem.
\begin{thm}
  Any PDA scheme can be represented by a rainbow scheme under the unified framework.
\end{thm}

\begin{rmk}
  In the delivery phase, if we take $m$ into consideration, then the transmission rate can be further reduced since $m\leq|\Phi|$.
\end{rmk}

\subsection{Construction from the union of disjoint subsets}
The first construction from \cite{2018ShangguanIT} regards users and packets as disjoint subsets of the ground set, respectively. More precisely, they set $$\bigg(K,F,\frac{M}{N},R\bigg)=\bigg(\binom{n}{a},\binom{n}{b},\frac{\binom{n}{b}-\binom{n-a}{b}}{\binom{n}{b}},\frac{\binom{n}{a+b}}{\binom{n}{b}}\bigg).$$
In particular, for $a=2,$ $n=\lambda a$ for constant $\lambda>1,$ this construction achieves $R=\lambda^{2}$ with $F=O(K^{-\frac{1}{4}}\cdot 2^{\sqrt{2K}H(\lambda^{-1})}),$ where $H(x)=-x\log_{2}{x}-(1-x)\log_{2}{(1-x)}$ for $0<x<1$ is the binary entropy function. Moreover, it is easy to check $R$ and $\frac{M}{N}$ are both constant and $F$ grows sub-exponentially with $K$ under such parameters.
This can actually be achieved by subset version of rainbow schemes as follows.
\begin{dfn}
 Let $\mathcal{A}=\binom{[n]}{a}$ be a collection of all $a$-element subsets of $[n]$ and $\mathcal{B}=\binom{[n]}{b}$ be a collection of all $b$-element subsets of $[n].$ Set the operation $\biguplus$ as set union operation $\bigcup.$ Suppose that $a$ and $b$ are positive integers with $a<b,$ and $n$ is large enough, then it is easy to see that
 \begin{equation*}
   \mathcal{C}=\mathcal{A}\biguplus\mathcal{B}=\binom{[n]}{b}\bigcup\binom{[n]}{b+1}\bigcup\cdots\bigcup\binom{[n]}{a+b}.
 \end{equation*}
 Let $\hat{\mathcal{C}}\subseteq\mathcal{C}$ be the collection of all $(a+b)$-element subsets of $[n],$ i.e. $\hat{\mathcal{C}}=\binom{[n]}{a+b}.$ It is easy to see $|\hat{\mathcal{C}}|=|\mathcal{C}|-o(|\mathcal{C}|).$
 We then color every element in $\hat{\mathcal{C}}$ using the proper coloring function $\Phi,$ and leave the elements in $\mathcal{C}\setminus\hat{\mathcal{C}}$ uncolored.
 \begin{enumerate}
    \item \textbf{Placement phase:} For uncolored elements $A\bigcup B$ in $\mathcal{A}\biguplus \mathcal{B},$ user $A$ caches the $B$-th packet of all files in the library.
    \item \textbf{Delivery phase:} The delivery is based on the coloring function $\Phi$. Suppose there are $s$ elements
    $A_{1}\bigcup B_{1},A_{2}\bigcup B_{2},\ldots,$
    $A_{s}\bigcup B_{s}$ receiving the same color from $\Phi$, then the server broadcasts the following XOR multiplexing of packets
          \begin{equation*}
            \bigoplus\limits_{1\leqslant i\leqslant s} W_{d_{A_{i}}}^{(B_{i})}.
          \end{equation*}
          The delivery phase consists of $|\Phi|$ packet transmissions, where $|\Phi|$ is the number of colors in $\Phi$.
 \end{enumerate}
\end{dfn}
It remains to discuss the properties of the proper coloring function $\Phi$. Let $A_{1},A_{2}\in\mathcal{A}$ and $B_{1},B_{2}\in\mathcal{B}$. Using the proof of Theorem \ref{thm:rainbow}, neither of the following will happen.
\begin{itemize}
  \item $A_{1}\bigcup B_{1}$ and $A_{1}\bigcup B_{2}$ cannot receive the same color from $\Phi.$ Also, $A_{2}\bigcup B_{1}$ and $A_{2}\bigcup B_{2}$ cannot receive the same color from $\Phi.$
  \item If $A_{1}\bigcup B_{1}$ and $A_{2}\bigcup B_{2}$ receive the same color from $\Phi,$ then both of $A_{1}\bigcup B_{2}$ and $A_{2}\bigcup B_{1}$ are uncolored.
\end{itemize}
To satisfy the above conditions, we can design the coloring functions with the following properties.
\begin{lem}
  Let $\Phi$ be the coloring function of elements in $\hat{\mathcal{C}}=\binom{[n]}{a+b}$ with the following properties:
  \begin{itemize}
    \item If $C_{1},C_{2}\in\hat{\mathcal{C}}$ and $|C_{1}\bigcap C_{2}|\geqslant a,$ then $\Phi(C_{1})\neq\Phi(C_{2}).$
    \item If $C_{1},C_{2},C_{3}\in\hat{\mathcal{C}}$ and $C_{i}\subseteq C_{j}\bigcup C_{k}$ for $i\neq j\neq k,$ then $C_{1},C_{2}$ and $C_{3}$ receive distinct colors.
  \end{itemize}
 Then $\Phi$ can be used in the above set system version of rainbow scheme.
\end{lem}

\begin{rmk}
  However, there is only a trivial coloring function $\Phi$ yet, that is, we color each element in $\hat{\mathcal{C}}=\binom{[n]}{a+b}$ via different colors. This trivial coloring function achieves the scheme of \cite{2018ShangguanIT}. Hence, any proper coloring function with fewer than $\binom{n}{a+b}$ colors will improve this construction.
\end{rmk}

\subsection{Ali-Niesen scheme}
  As far as we know, Ali-Niesen scheme \cite{2014AliIT} is the first coded caching scheme that kickstarted the research of coded caching in general. Recall the parameters in Maddah Ali-Niesen scheme as
  $$\bigg(K,F,\frac{M}{N},R\bigg)=\bigg(K,\binom{K}{\frac{KM}{N}},\frac{M}{N},\frac{K(1-\frac{M}{N})}{\frac{KM}{N}+1}\bigg).$$
Using the generalized rainbow scheme, we can set $\mathcal{A}=[n]$ and $\mathcal{B}=\binom{[n]}{t},$ where $t=\frac{KM}{N}.$ Using the similar analysis, it suffices to design the coloring functions with properties as follows.
\begin{lem}
  Let $t=\frac{KM}{N}$ and $\Phi$ be the coloring function of elements in $\hat{\mathcal{C}}=\binom{[n]}{t+1}$ such that
  \begin{itemize}
    \item If $C_{1},C_{2}\in\hat{\mathcal{C}}$ and $C_{1}\bigcap C_{2}\neq\emptyset,$ then $\Phi(C_{1})\neq\Phi(C_{2}).$
    \item If $C_{1},C_{2},C_{3}\in\hat{\mathcal{C}}$ and $C_{i}\subseteq C_{j}\bigcup C_{k}$ for $i\neq j\neq k,$ then $C_{1},C_{2},C_{3}$ receive distinct colors.
  \end{itemize}
 Then $\Phi$ can be used in the Ali-Niesen type rainbow scheme.
\end{lem}

\subsection{Tang-Ramamoorthy scheme}

In the Tang-Ramamoorthy scheme \cite{2017TANGLIISIT}, the subpacketization is exponentially smaller than that of the previous scheme but with some minor loss in rate (up to a constant factor). However, $F=e^{f(\frac{N}{M})\cdot K}$ for some function $f(\cdot)$.
We show an example which can also be achieved by our generalized rainbow scheme. Assume that there exists a generator matrix $G$ of an $(n,k)$ linear block code over field $\mathbb{F}_{q}$ which has the following properties:
\begin{enumerate}
  \item \textbf{Divisibility:} $(k+1)\mid n,$
  \item \textbf{Rank property:} For every contiguous set of $k+1$ columns, every subset of $k$-columns on this $k+1$ subset has full rank.
\end{enumerate}
Then we can obtain $q^k$ codewords of length $n$.

We use $\mathcal{A}$ to denote all the pairs $(a,a')$, where $a\in[n]$ and $a'\in\{0,1,\ldots,q-1\}$.
Each codeword $b_1b_2\ldots b_n$ corresponds to the set $\{(1,b_1),(2,b_2),\ldots,(n,b_n)\}$. Therefore, all of the $q^k$ codewords correspond to $q^k$ sets, which form the family $\mathcal{B}$.
\begin{dfn}
Let $\mathcal{A}$ and $\mathcal{B}$ be the families defined above. Let $\biguplus$ be the set union operation $\bigcup$ and $\mathcal{C}=\mathcal{A}\cup\mathcal{B}$. It is easy to show that for any $C\in\mathcal{C}$, $C$ contains exactly $n$ or $n+1$ pairs. Let $\hat{\mathcal{C}}\subseteq\mathcal{C}$ be the collection of elements in $\mathcal{C}$ which contains exactly $n+1$ pairs. It is easy to see $|\hat{\mathcal{C}}|=|\mathcal{C}|-o(|\mathcal{C}|)$. We then color every element using the proper coloring function, and leave the elements in $\mathcal{C}\setminus\hat{\mathcal{C}}$ uncolored. The placement phase and delivery phase are the same as those in Definition \ref{def:generalizedScheme}.
\end{dfn}

Now, we give a coloring function when $n=k+1$. Any $\{C_1,\dots,C_{k+1}\}\subseteq\hat{\mathcal{C}}$ satisfying
\begin{flalign}\label{rainbow_pro}
|C_i\cap C_j|=n,~\text{for any }i\ne j\in[k+1],
\end{flalign}
forms a color class. Next we explain such a coloring function in detail.

For each sequence $s_1s_2\ldots s_n$ which is not a codeword, its corresponding set $\{(1,s_1),\ldots,(n,s_{n})\}\notin\mathcal{B}$. Due to the rank property, for each $j\in[n]$, there is a unique codeword $b_1b_2\ldots b_n$ such that $b_i=s_i$ for all $i\in[n]\setminus\{j\}$. Therefore, for $j\in[n]$, there is a unique element $B_j$ in $\mathcal{B}$ such that $\{(k,s_k)|k\in[n]\setminus j\}\subseteq B_{j}$. We color $B_j\cup(j,s_j)\in\hat{\mathcal{C}}$ for every $j\in[k+1]$ with the same color, thus the condition (\ref{rainbow_pro}) is satisfied, where the set of $n$ common pairs corresponds to the sequence $s_1s_2\ldots s_n$.

\begin{example}
  Suppose $q=2$, $k=2$, $n=k+1=3$, we have $K=6$, $F=4$. The generator matrix $G$ of $(3,2)$ block code is $$\left[\begin{array}{ccc}1&0&1\\0&1&1\end{array}\right].$$
  Then $\mathcal{A}=\{(1,0),(1,1),(2,0),(2,1),(3,0),(3,1)\}$, $\mathcal{B}=\{B_1,B_2,B_3,B_4\}$, where
  \begin{flalign*}
 B_1=\{(1,0),(2,0),(3,0)\},\\
 B_2=\{(1,0),(2,1),(3,1)\},\\
 B_3=\{(1,1),(2,0),(3,1)\},\\
 B_4=\{(1,1),(2,1),(3,0)\}.
  \end{flalign*}
  Using the procedure above, we can color the elements in $\hat{C}$ with $4$ colors.
\begin{flalign*}
 c_1=\{&\{(1,1),(1,0),(2,0),(3,0)\},\\
 &\{(2,0),(1,1),(2,1),(3,0)\},\\
 &\{(3,0),(1,1),(2,0),(3,1)\}\};\\
 c_2=\{&\{(2,1),(1,0),(2,0),(3,0)\},\\
 &\{(1,0),(1,1),(2,1),(3,0)\},\\
 &\{(3,0),(1,0),(2,1),(3,1)\}\};\\
 c_3=\{&\{(3,1),(1,0),(2,0),(3,0)\},\\
 &\{(1,0),(1,1),(2,0),(3,1)\},\\
 &\{(2,0),(1,0),(2,1),(3,1)\}\};\\
 c_4=\{&\{(3,1),(1,1),(2,1),(3,0)\},\\
 &\{(1,1),(1,0),(2,1),(3,1)\},\\
 &\{(2,1),(1,1),(2,0),(3,1)\}\}.
\end{flalign*}
Finally, it achieves a $(K,F,\frac{M}{N},R)=(6,4,\frac{1}{2},1)$ centralized coded caching scheme.
\end{example}

\section{New rainbow schemes for coded caching}\label{rainbow}

\subsection{New rainbow schemes}
In this section, we introduce our new scheme for the centralized coded caching problem under the rainbow framework in the previous section. Suppose there are $K$ users served through a noiseless broadcast channel by an agent who has access to $N$ distinct files from a library. Every user is equipped with a local cache of size $M.$ The key problem is to design the \emph{placement phase} where the user caches file packets from the library under the cache constraint and the \emph{delivery phase} where the user reveals his own demand so that all demands should be satisfied with at most $R$ file transmissions.

Let $\mathcal{W}=\{W_{1},W_{2},\ldots,W_{N}\}$ be a library of $N$ files. Let $(W_{i}^{(1)},W_{i}^{(2)},\ldots,W_{i}^{(F)})\in \mathbb{F}^{F\times 1}$ be a vector of length $F$ over some field $\mathbb{F}$ representing file $W_{i}.$ Then we recall the $(R,K,M,N,F)$ centralized coded caching scheme as follows.

\begin{dfn}
  Every file $W_{i}$ in the library is divided into $F$ packets for $1\leqslant i\leqslant N.$ An $(R,K,M,N,F)$ centralized coded caching scheme consists of:
  \begin{enumerate}
    \item A family of subsets $\{W_{i,j}\}_{i\in [N],j\in [K]},$ where $W_{i,j}\subseteq [K]$ is the set of user caches where the $i$-th packet of file $j$ is stored. Moreover, each user can cache at most $MF$ file packets in placement phase.
    \item A set of user demands $\textbf{d}=(d_{1},d_{2},\ldots,d_{K})$ arising from the library, where $d_{i}\in [N]$ is the index of the requested file of the user $k.$ The transmission function $\phi(W_{d_{1}},W_{d_{2}},\ldots,W_{d_{K}})\rightarrow \mathbb{F}^{RF\time 1}$ for some field $\mathbb{F}$ such that every user $s$ can decode their demanded files $W_{d_{s}}$ via $\phi$ and the cache content available.
    \item For any demand pattern among the users arising from the library, the total number of file transmission can be at most $R.$
  \end{enumerate}
\end{dfn}

\begin{dfn}[Rainbow coded caching scheme]\label{def:basicRainbow}
  Let $K=m$ be an integer. Let $A\subseteq [2m]$ be an $(\alpha,\beta)$-only rainbow $3$-APs set. More precisely, $A\subseteq [2m]$ is a set of size at least $2m-(2m)^{\alpha},$ and let $\chi$ be a coloring of $A$ with at most $(2m)^{\beta}$ colors such that every $3$-AP in $A$ is rainbow. Let $A_{1}=A_{2}=[m],$ we consider the following sum set
  \begin{equation*}
    A_{1}+A_{2}=\{x+y:x\in A_{1}, y\in A_{2}\}.
  \end{equation*}
  We then color the pairs $(x,y)\in A_{1}\times A_{2}$ as
  \[\Psi((x,y))=\left
  \{\begin{array}{ll}

uncolored,&\text{$x+y\notin A$},\\

(x-y,\chi(x+y)),&

\text{$x+y\in A$}.

\end{array}\right.\]
    Then we describe the placement phase and delivery phase with assistance of the above colored sum set. In this scheme, every file in the library is split into $F=K$ packets.
    \begin{enumerate}
      \item \textbf{Placement phase:} For uncolored elements $x+y$ in $A_{1}+A_{2},$ user $x$ caches the $y$-th packet of all files in the library.
      \item \textbf{Delivery phase:} The delivery is based on the coloring function $\Psi$. Suppose that there are $s$ elements $x_{1}+y_{1},x_{2}+y_{2},\ldots,x_{s}+y_{s}$ receiving the same color from $\Psi$, then the server broadcasts the following XOR multiplexing of packets
          \begin{equation*}
            \bigoplus\limits_{1\leqslant i\leqslant s} W_{d_{x_{i}}}^{(y_{i})}.
          \end{equation*}
          The delivery phase consists of $|\Psi|$ packet transmissions, where $|\Psi|$ is the number of colors in $\Psi$.
    \end{enumerate}
  \end{dfn}

\begin{rmk}
  Note that in this new scheme, we omit the constant $m$ defined in the framework and send messages according to their colors directly.
\end{rmk}

Next we show that our rainbow scheme is a centralized coded caching scheme.
\begin{thm}\label{thm:rainbow}
  The $(K,\alpha,\beta,\Psi)$ rainbow scheme is an $(R=\frac{|\Psi|}{F},K,M,N,F=K)$ coded caching scheme.
\end{thm}
\begin{proof}
  In our rainbow scheme, the number of packets per file $F$ is equal to the number of users $K.$ Our first task is to verify that the cache constraint of every user is satisfied. Note that for every user $x\in [m],$ there are at most $(2m)^{\alpha}$ elements $b\in [2m]\setminus A,$ such that $x+y=b.$ This indicates that each user caches at most $(2m)^{\alpha}N\leqslant MF$ file packets.

Next we will show that our rainbow scheme satisfies any kind of user demands $\textbf{d}=(d_{1},d_{2},\ldots,d_{K})$ arising from the library, where $d_{i}\in [N]$ is the index of the requested file of the user $i.$ After coloring each element of sum set $A_{1}+A_{2}$ using function $\Psi,$ we consider some color class $\mathcal{C}_{j}$ consisting of $c_{j}$ elements
\begin{equation*}
  x_{1}+y_{1},x_{2}+y_{2},\ldots,x_{c_{j}}+y_{c_{j}},
\end{equation*}
and the corresponding XOR transmission consisting of $c_{j}$ packets:
\begin{equation*}
 \bigoplus\limits_{1\leqslant i\leqslant c_{j}}W_{d_{x_{i}}}^{(y_{i})}.
\end{equation*}
Then we analyze the decoding algorithm for each user. For a user $x\in [m]$ requesting a certain file $W_{d_{x}},$ he has cached the set of packets $\{W_{s,d_{x}}:\text{$s+x$ is uncolored}\}$ in the placement phase. Hence, to decode the requested file $W_{d_{x}},$ it suffices to obtain the uncached packets. We just need to show the following result.
\begin{clm}\label{claim:s=t}
Let $\mathcal{C}_{j}$ be some color class consisting of $c_{j}$ elements
\begin{equation*}
  x_{1}+y_{1},x_{2}+y_{2},\ldots,x_{c_{j}}+y_{c_{j}}.
\end{equation*}
  Then for each $1\leqslant s,t\leqslant c_{j}$, $x_{s}+y_{t}\in A$ if and only if $s=t.$
\end{clm}
\begin{proof}[Proof of Claim \ref{claim:s=t}]
  Trivially, $x_{s}+y_{t}\in A$ if $s=t$ by the definition. On the other hand, if $s\neq t,$ suppose $x_{s}+y_{t}\in A,$ the element $x_{s}+y_{t}$ will receive a color from function $\Psi.$ Without loss of generality, let $\Psi(x_{s}+y_{s})=\psi(x_{t}+y_{t})=\Psi_{1}$ and $\psi(x_{s}+y_{t})=\Psi_{2}.$ Let $\Delta=x_{s}-y_{s}=x_{t}-y_{t},$ then we can write $x_{s}+y_{t}$ as
  \begin{equation*}
    x_{s}+y_{t}=\frac{(2x_{s}-\Delta)+(2y_{t}+\Delta)}{2}=\frac{(x_{s}+y_{s})+(x_{t}+y_{t})}{2}.
  \end{equation*}
Note that $x_{s}+y_{s},$ $x_{s}+y_{t}$ and $x_{t}+y_{t}$ form a $3$-AP in $A.$ However, $x_{s}+y_{s}$ and $x_{t}+y_{t}$ are in the same color class, which implies $\chi(x_{s}+y_{s})=\chi(x_{t}+y_{t})$ by the definition of $\Psi.$ That is impossible since every $3$-AP in $A$ is rainbow, then the claim follows.
\end{proof}
By Claim \ref{claim:s=t}, it holds that user $x_{s}$ knows all the packets $\{W_{d_{x_{i}}}^{(y_{i})}: 1\leqslant i\leqslant c_{j},i\neq s\}$ in his cache at the placement phase. Then the unknown packets can be easily obtained by substraction operation. Every user will recover his requested file by this decoding algorithm, therefore the rainbow scheme works.
Finally, it is easy to see that the number of colors used by $\Psi$ is at most $(2m)^{1+\beta}.$ This completes the proof of Theorem \ref{thm:rainbow}.
\end{proof}

Next we give an example to show the practicality of our coded caching scheme.
\begin{example}
  Suppose $K=F=4,$ we color the elements in $\{1,2,\ldots,8\}$ as follows.

  \[\Theta(x)=\left
  \{\begin{array}{ll}
uncolored,&\text{$x=1,5$},\\
a,&\text{$x=2,8$},\\
b,&\text{$x=3,7$},\\
c,&\text{$x=4,6$}.
\end{array}\right.\]

Using the function $\Psi$ in Definition~\ref{def:basicRainbow}, we present a table as follows.
\begin{table}[!htbp]
\caption{$K=F=4$ rainbow coded caching scheme}\label{table:example}
\centering
\begin{tabular}{|c|c|c|c|c|}
\hline
\diagbox{$K$}{$F$}&$1$&$2$&$3$&$4$\\
\hline
$1$&$(0,a)$&$(-1,b)$&$(-2,c)$&uncolored\\
\hline
$2$&$(1,b)$&$(0,c)$&uncolored&$(-2,c)$\\
\hline
$3$&$(2,c)$&uncolored&$(0,c)$&$(-1,b)$\\
\hline
$4$&uncolored&$(2,c)$&$(1,b)$&$(0,a)$\\
\hline
\end{tabular}
\end{table}
As we can see in Table~\ref{table:example}, every uncolored pair $(x,y)$ represents a caching action in the placement phase and every color class corresponds to a transmission in delivery phase. Finally, it achieves a $(K,F,\frac{M}{N},R)=(4,4,\frac{1}{4},6)$ centralized coded caching scheme.
\end{example}

Obviously, the limitation showed in \cite{2018ShangguanIT} indicates the following result.
\begin{thm}
  The only rainbow $3$-APs set with $n-n^{\alpha}$ elements and only $O(1)$ colors does not exist for any $0<\alpha<1.$
\end{thm}

\subsection{Schemes taking $m$ into consideration}
Note that in the above rainbow scheme, we first color the non-cached subfiles such that special structure in $\hat{\mathcal{C}}$ is rainbow and then deliver messages depending on their colors. If different colors can be sent together, then the transmission load can be further reduced. In Example~\ref{egcaching}, making use of the results of the index coding problem, we calculate a constant $m$, and use an $m\times|\Phi|$ MDS matrix to combine different colors.

\begin{rmk}
  In Example~\ref{egcaching}, if we deliver the subfiles according to different colors, we have to send messages at 4 times. But if we use the new delivery scheme, we only need to send messages at $3$ times and for each time, the subfiles corresponding to $2$ colors are used.
\end{rmk}

For general case, let $\mathcal{A}=[n]$ and $\mathcal{B}\subset\binom{[n]}{t}$. Set the operator as set union $\cup$.
Then $\mathcal{C}=\mathcal{A}\cup\mathcal{B}$ contains some $t$-tuples and $(t+1)$-tuples of $[n]$. Define $\hat{\mathcal{C}}=\mathcal{C}\cap\binom{[n]}{t+1}$, and the coloring function $\phi$ over $\hat{\mathcal{C}}$ must satisfy the following rainbow conditions.
\begin{itemize}
    \item If $C_{1},C_{2}\in\hat{\mathcal{C}}$ and $C_{1}\bigcap C_{2}\neq\emptyset,$ then $\Phi(C_{1})\neq\Phi(C_{2}).$
    \item If $C_{1},C_{2},C_{3}\in\hat{\mathcal{C}}$ and $C_{i}\subseteq C_{j}\bigcup C_{k}$ for $i\neq j\neq k,$ then $C_{1},C_{2}$ and $C_{3}$ receive different colors.
\end{itemize}
The coloring function $\Phi$ over $\{(A,B):\ A\cup B\in\hat{\mathcal{C}}\}$ is defined as
\begin{equation*}
  \Phi(A,B)=\phi(A\cup B).
\end{equation*}
Suppose that $(A_1,B_1),(A_2,B_2),\ldots,(A_s,B_s)$ are assigned with the same color $c$, which implies that $A_i\cup B_j\notin\hat{\mathcal{C}}, \forall i\neq j$. Thus user $A_i$ can decode $W_{d_{A_i}}^{(B_i)}$ from $W_c=\bigoplus_{j\in[s]} W_{d_{A_j}}^{(B_j)}$.
For each pair $(A,B)$ with $A\cup B\in\hat{\mathcal{C}}$, define a constant
\begin{equation*}
  m(A,B)=\#\{\Phi(A'\cup B'):\ A'=A\ or\ B'\cup A\in\hat{\mathcal{C}}\},
\end{equation*}
where $(A',B')=(A,B)$ is allowed. Then define
\begin{equation*}
  m=\max\{\{m(A,B):A\cup B\in\hat{\mathcal{C}}\}\cup\{|\Phi|-1\}\}.
\end{equation*}
Let $P$ be an $m\times|\Phi|$ MDS matrix,
\begin{equation*}
  P=\begin{pmatrix}
      p_1 & p_2 & \cdots & p_\Phi
    \end{pmatrix}.
\end{equation*}
During the delivery phase, the server sends
\begin{equation*}
  P\cdot(W_{c_1}, W_{c_2}, \cdots, W_{c_{|\Phi|}})^T.
\end{equation*}
Based on the above construction, we have the following theorem.
\begin{thm}
  The rainbow framework with new delivery scheme described above is an $(R=\frac{m}{F}, K=|\mathcal{A}|, F=|\mathcal{B}|,M,N)$ coded caching scheme.
\end{thm}
\begin{pf}
  It suffices to prove the solvability for each user $A_i$. From the definition of $\hat{\mathcal{C}}$, we know that user $A_i$ caches the subfile $B_j$ if  $A_i\cup B_j\notin\hat{\mathcal{C}}$, and $A_i$ can decode the subfile $B_j$, which is not cached by the user, from $W_{c_1}$ if $\Phi(A_i,B_j)=c_1$. From the definition of $m(A_i,B_j)$, it is not difficult to see that every subfile $W_{d_{A_m}}^{(B_n)}$ with $A_m\cup B_n\in\hat{\mathcal{C}}$ which is cached by user $A_i$ is not contained in the set
  \begin{equation*}
    \{(A',B'):\ A'=A_i\ or\ B'\cup A_i\in\hat{\mathcal{C}}\}.
  \end{equation*}
  Thus, for user $A_i$, after deleting the subfiles he has cached, we have
  \begin{equation*}
    (p_1,p_2,\cdots,p_{m(A_i,B_j)})\cdot(W_{c_1}, W_{c_2}, \cdots, W_{c_{m(A_i,B_j)}})^T.
  \end{equation*}
  Since $P$ is an MDS matrix and $m\geq m(A_i,B_j)$, user $A_i$ can decode $W_{c_1}$ and further decode $W_{d_{A_i}}^{(B_j)}$. The similar approach can be used for any $(A_i, B_j)$, thus the solvability of the new delivery scheme is proved.
\end{pf}

\section{Application to distributed computing}\label{distributed}
In this section, we present the application of our rainbow framework in the distributed computing problem. We start with an illustrative example.
\begin{example}\label{egdistributed}
  Let $\mathcal{A}=[4],\mathcal{B}=\{12,23,34,41\}$. Suppose there are $4$ files $\{W_B:B\in\mathcal{B}\}$, $4$ work nodes $\{K_a:a\in\mathcal{A}\}$ and $4$ functions $\{\phi_a:a\in\mathcal{A}\}$. In the map phase, work node $K_a$ is assigned file $W_B$ if $a\in B$, and computes the intermediate values $\{v_{a',B}:a'\in\mathcal{A}, a\in B\}$. For example, work node $K_1$ gets files $W_{12},W_{41}$ and computes all intermediate values based on these two files, i.e. $\{v_{a',12},v_{a',41}\}$ for all $a'\in\mathcal{A}$. In the shuffle phase, work node $K_a$ needs all the intermediate values to compute $\phi_a$. In this case,
  \begin{enumerate}
    \item $K_1$ has $v_{a',12},v_{a',41}$ and needs $v_{1,23},v_{1,34}$;
    \item $K_2$ has $v_{a',12},v_{a',23}$ and needs $v_{2,41},v_{2,34}$;
    \item $K_3$ has $v_{a',23},v_{a',34}$ and needs $v_{3,12},v_{3,41}$;
    \item $K_4$ has $v_{a',34},v_{a',41}$ and needs $v_{4,23},v_{4,12}$.
  \end{enumerate}
  From the rainbow framework above, there are $4$ color groups
  \begin{equation*}
    \{v_{1,23},v_{3,12}\},\{v_{1,34},v_{3,41}\},\{v_{2,41},v_{4,12}\},\{v_{2,34},v_{4,23}\}.
  \end{equation*}
  Then, during the shuffle phase,
  \begin{enumerate}
    \item $K_1$ sends $v_{2,41}\bigoplus v_{4,12};$
    \item $K_2$ sends $v_{1,23}\bigoplus v_{3,12};$
    \item $K_3$ sends $v_{2,34}\bigoplus v_{4,23};$
    \item $K_4$ sends $v_{1,34}\bigoplus v_{3,41}.$
  \end{enumerate}
  It is easy to check that each work node can get all the intermediate values it needs to complete the computation. In this example, there are $4\times4=16$ intermediate values and in the shuffle phase, there are only $4$ transmissions, each of size equal to one intermediate value. Thus, the communication load $L=\frac{4}{16}=\frac{1}{4}$.
\end{example}

\begin{rmk}
  Note that in Example \ref{egdistributed}, the computation load $r=2$ and the number of work nodes $K=4$. From the relationship of the communication load and computation load in \cite{2018ITdistributed}, $L\geq\frac{1}{r}(1-\frac{r}{K})=\frac{1}{4}$, which shows our rainbow scheme also attains the optimal transmission while only needs to separate the original file into $4$ parts which appeared in \cite{2018ITdistributed}, the number of files has to be the multiple of $\binom{K}{r}=6$. Comparing with the previous scheme constructed from PDAs and resolvable designs \cite{2020YanDistributed},\cite{2020KKRA}, our rainbow scheme has a lower communication load. In detail, for both PDA scheme and resolvable design scheme, when $r=2, K=4$, their communication load $L=\frac{1}{2}>\frac{1}{4}$.
\end{rmk}

\begin{rmk}
  The main difference between the rainbow scheme and the previous schemes is that we do not require multicast groups. In \cite{2020KKRA},\cite{2018ITdistributed},\cite{2020YanDistributed}, all work nodes send intermediate values according to the multicast groups. In each multicast group, after deleting any one work node in the group, the remaining nodes share only one file which is needed by the deleted node. Thus, this common file is partitioned into pieces and sent by all remaining nodes. But in the rainbow scheme for coded caching shown in Example \ref{egcaching}, the transmissions should be
  \begin{align*}
      & v_{1,23}\bigoplus v_{3,12}\bigoplus v_{2,34}\bigoplus v_{4,23}, \\
      & v_{2,14}\bigoplus v_{4,12}\bigoplus v_{2,34}\bigoplus v_{4,23}, \\
      & v_{1,34}\bigoplus v_{3,14}\bigoplus v_{2,34}\bigoplus v_{4,23}.
  \end{align*}
  We can partition each transmission into $2$ parts, each part can be sent by only one work node. Therefore, we obtain the shuffle phase in Example \ref{egdistributed}.
\end{rmk}

We modify the rainbow framework in Section~\ref{framework} to be applicable to distributed computing problem as follows.
\begin{dfn}
  Let $K=|\mathcal{A}|,$ $N=|\mathcal{B}|$ and $K|Q$, define $\mathcal{C}=\mathcal{A}\biguplus\mathcal{B}.$ Let the $\sigma$-type structure, $\hat{\mathcal{C}},$ and the coloring functions $\phi$ and $\Phi$ be the same as those given in Definition~\ref{def:generalizedScheme}.
  Then we describe the MapReduce scheme with assistance of the above colored subset $\hat{\mathcal{C}}.$
  \begin{enumerate}
    \item \textbf{Map phase:} Work node $a$ caches $\{W_b: a\biguplus b\notin\hat{\mathcal{C}}\},$ and computes intermediate values $\{v_{q,b}:q\in[Q],b\in\mathcal{B}\}$.
    \item \textbf{Shuffle phase:} Work node $a$ is assigned to compute functions $\mathcal{Q}_a$ with size $\frac{Q}{K}$.
    The communication between nodes is based on the coloring function $\Phi$. Denote the delivery in Definition~\ref{def:generalizedScheme} as
    \begin{equation*}
      P\cdot(v_{c_1}, v_{c_2}, \cdots, v_{c_{|\Phi|}})^T=P\cdot V,
    \end{equation*}
    where $P$ is an $m\times|\Phi|$ MDS array, and $v_{c_i}$ is the $XOR$ sum of all the intermediate values that have color $c_i$.
    Suppose $P\cdot V$ can be decomposed into a linear combination of $v_1,v_2,\cdots,v_{m'}$, where each $v_i$ can be contained in one work node $i$. Then the final transmission is that node $i$ sends $v_i$ for $i\in[m']$.
    \item \textbf{Reduce phase:} After transmission, each work node gets the information it needs and computes functions which are assigned to it.
  \end{enumerate}
\end{dfn}

The main idea of this new shuffle phase is that we separate the transmission of the rainbow framework in Section~\ref{framework} into several pieces such that each can be sent by one user. Note that for a $g$-regular PDA scheme in distributed computing, that is, each color group (or multicast group) contains exactly $g$ work nodes, the transmission load in each group is $\frac{g}{g-1}$, therefore, the whole communication load is $\frac{g}{g-1}|\Phi|$. In our scheme, if the transmission of the rainbow framework is partitioned into $m'$ pieces, then the communication load is $m'$. Therefore, when $m=|\Phi|$, unless each color scheme can be sent by one work node, the original PDA scheme performs better. However, for some special cases where $m< |\Phi|$, $m'$ can be smaller than the previous scheme, which means for the same number of files, our scheme can have a lower transmission load. We present a class of schemes with $m'<\frac{g}{g-1}|\Phi|$ under the same number of files.

\begin{example}
  Consider $n\geq 5$ (where $n=4$ is considered in Example~\ref{egdistributed}). Let $\mathcal{A}=[n],$ $\mathcal{B}=\{[i:i+n-3]\mod n:i\in[n]\}$ and $\hat{\mathcal{C}}=\mathcal{C}\cap\binom{[n]}{n-1}$. Suppose there are $N=\mathcal{B}$ files, $K=\mathcal{A}$ work nodes and $Q$ functions, for convenience, we assume $Q=K$. Coloring all elements in $\hat{\mathcal{C}}$ with different colors, thus $|\Phi|=n$. For each pair $(a,b)$,
  \begin{equation*}
    m(a,b)=\#\{\Phi(a'\cup b'):b'\cup a\in\hat{\mathcal{C}}\}= 3,
  \end{equation*}
  therefore, $m=|\Phi|-1=n-1$. Select the $(n-1)\times n$ MDS matrix $P$ as follows
  \begin{equation*}
    \begin{pmatrix}
      1 & 1 & 0 & \cdots & 0 & 0\\
      0 & 1 & 1 & \cdots & 0 & 0\\
      \vdots & \vdots & \vdots & \ddots & \vdots & \vdots\\
      0 & 0 & 0 & \cdots & 1 & 1
    \end{pmatrix},
  \end{equation*}
  the $i$-th column of $P$ corresponds to the color of $[i:i+n-2]\in\hat{\mathcal{C}}$. Note that work node $i+2$ contains all intermediate values of $\{v_{a,b}:a\cup b=[i:i+n-2]\}$ and $\{v_{a,b}:a\cup b=[i+1:i+n-1]\}$ since $n\geq5$. Therefore in the shuffle phase, work node $i+1$ sends the $i$th row of $P\times(v_{[1:n-1]},v_{[2:n]},\cdots,v_{[n:n-2]})$ where $v_{[i:i+n-2]}=\bigoplus_{a\cup b=[i:i+n-2]}v_{a,b}$.

  In this scheme, there are $m'=n-1<n=|\Phi|$ work nodes that need to send a message, therefore the final transmission load $L=\frac{n-1}{n^2}$.
\end{example}

\begin{rmk}
  In the above example, we obtain a scheme with $L=\frac{n-1}{n^2}$ and $K=N=n$, $r=n-2$. Comparing with existing schemes, our scheme has $N=\Theta(K)$ and $L/L^*=\Theta(K)$.
\end{rmk}

When $m'\geq\frac{g}{g-1}|\Phi|$, we still use the same multicast group idea as the PDA schemes in \cite{2020YanDistributed}.

\section{Conclusion}\label{conclusion}

In this paper, we investigate two highly-related network communication problems which are known as coded caching and distributed computing respectively. Motivated by the study of only rainbow $3$-APs set, we propose a generalized rainbow framework which can be applied to both problems.
We observe that any PDA scheme can be represented by a rainbow scheme under our framework, and several existing works can also be included in the framework.
Our rainbow framework build bridges between combinatorial objects and coded caching problems. For any given $\mathcal{A}$, $\mathcal{B}$ and the operation $\biguplus$, we can obtain a coded caching scheme by selecting a suitable $\sigma$-type structure and the coloring function. The freedom of choosing the structure and coloring function enables us to connect more combinatorial objects with coded caching. Moreover, using the idea of the index coding problem, our framework can further reduce the transmission load and obtain some schemes which cannot be represented by a PDA.

Next, based on the study of the only rainbow $3$-term arithmetic progressions set, we offer a coded caching scheme with linear subpacketization and near constant rate. For several existing works, we propose the corresponding coloring models and we do hope it will be helpful to solve the following problem by designing proper coloring function.
\begin{ques}
  Let $\frac{M}{N}$ and $R$ be both constants, prove or disprove the existence of centralized coded caching schemes such that $F$ grows polynomially with $K.$
\end{ques}

At last, we investigate the application of this rainbow framework in distributed computing problems. For some special cases, we propose a new shuffle  phase, in which work nodes do not need to send messages based on multicast groups. This new communication scheme may bring some advantages comparing with the PDA schemes under the same number of files.
Due to the connection between coded caching and distributed computing, the number of files in distributed computing is equivalent to the number of subpaketizations of coded caching. Therefore, a similar question can be asked as following.
\begin{ques}
  Let $r$ be a constant. Suppose $N$ cannot be divided by $\binom{K}{r+1}$ or is small than $\exp(K)$, construct schemes which can attain the optimal or suboptimal communication load $L$.
\end{ques}

\bibliographystyle{IEEEtranS}
\bibliography{reference}
\end{document}